\documentclass[twoside,11pt]{article}

\usepackage{blindtext}

\usepackage{jmlr2e} 

\usepackage{comment}

\usepackage{ulem}

\usepackage{microtype}
\usepackage{bm}
\usepackage{amsbsy}
\usepackage{multirow}

\usepackage{xcolor}

\usepackage{hyperref}       
\usepackage{url}            
\usepackage{booktabs}       
\usepackage{amsfonts}       
\usepackage{nicefrac}       
\usepackage{microtype}      
\usepackage{balance}

\usepackage{epsfig}
\usepackage{epstopdf}

\usepackage{mathtools}
\usepackage{algorithm}
\usepackage{algorithmic}

\usepackage{makecell}

\newcommand{\Osymbol}{{\mathcal O}}
\newcommand{\BO}[1]{\Osymbol\left(#1\right)}

\newcommand{\E}[1]{\mathbb{E} \left[#1\right]}
\newcommand{\Var}[1]{\mathrm{Var}\left[#1\right]}

\newcommand{\bX} {\mathbf{X}}

\newcommand{\bw} {\mathbf w}

\newcommand{\bx} {\mathbf x}
\newcommand{\by} {\mathbf y}

\newcommand{\bu} {\mathbf u}
\newcommand{\bv} {\mathbf v}

\renewcommand{\Pr}[1]{\textrm{\bf Pr}\left[#1\right]}

\newcommand{\mRd} {\mathbb{R}^d}
\newcommand{\mRD} {\mathbb{R}^D}
\newcommand{\mR} {\mathbb{R}}

\newcommand{\bY} {\mathbf{Y}}

\newcommand{\bxp} {\mathbf x^{(p)}}
\newcommand{\byp} {\mathbf y^{(p)}}

\newcommand{\bCx} {{\mathbf {Cx}}}

\newcommand{\bCy} {\mathbf {Cy}}

\newcommand{\bCxone} {\mathbf {C_1x}}
\newcommand{\bCxtwo} {\mathbf {C_2x}}

\newcommand{\dotxy}{\left\langle \bx, \by \right\rangle} 

\jmlrheading{x}{20xx}{x-xx}{x/xx}{xx/xx}{meila00a}{xxxx and xxxx}

\ShortHeadings{Tensor Sketch}{Pham and Pagh}
\firstpageno{1}

\begin{document}

\title{Tensor Sketch:\\Fast and Scalable Polynomial Kernel Approximation}

\author{\name Ninh Pham \email ninh.pham@auckland.ac.nz\\ 
\addr School of Computer Science\\ 
University of Auckland\\ 
Auckland, New Zealand
\AND 
\name Rasmus Pagh \email pagh@di.ku.dk\\ 
\addr Computer Science Department\\
University of Copenhagen\\ 
Copenhagen, Denmark} 

\editor{xxxx and xxxx}

\maketitle

\begin{abstract}

Approximation of non-linear kernels using random feature maps has become a powerful technique for scaling kernel methods to large datasets.
We propose \textit{Tensor Sketch}, an efficient random feature map for approximating polynomial kernels. 
Given $n$ training samples in $\mathbb{R}^d$ Tensor Sketch computes low-dimensional embeddings in $\mathbb{R}^D$ in time $\mathcal{O}\left( n(d+D \log{D}) \right)$ making it well-suited for high-dimensional and large-scale settings. 
We provide theoretical guarantees on the approximation error, ensuring the fidelity of the resulting kernel function estimates. 
We also discuss extensions and highlight applications where Tensor Sketch serves as a central computational tool.

\end{abstract}


\section{Introduction}\label{sec:introduction}

Kernel machines such as Support Vector Machines (SVMs)~\citep{Scholkopf02}, are powerful tools for a wide range of machine learning and data mining tasks. 
A key strength of kernel methods lies in their ability to capture non-linear structure in data through the use of kernel functions. These functions \textit{implicitly} map data from the original space to a high-dimensional feature space, where each coordinate corresponds to a feature of the input vectors. 
In this kernel space, many standard learning algorithms operate using only pairwise inner products, without explicitly computing the mapped coordinates. 
This implicit approach not only reduces computational overhead but also enables kernel methods to handle diverse data types, including both numeric and symbolic inputs, in a unified framework.

While kernel methods have achieved considerable success across a wide range of data analysis tasks~\citep{Taylor04}, their scalability remains a significant limitation. Kernel-based learning algorithms typically suffer from high computational and memory costs, with direct methods often requiring cubic time and quadratic space in the number of training samples~\citep{Scholkopf02}. 
This issue becomes increasingly problematic in modern machine learning applications that rely on large-scale datasets. In high-dimensional domains such as text, where data is often represented using sparse bag-of-words vectors, linear models -- particularly linear SVMs -- have demonstrated strong empirical performance~\citep{Joachims98,Lewis04}.
Consequently, there has been substantial research on designing training algorithms for linear SVMs that scale linearly with the number of training examples~\citep{Joachims06,Pegasos,Liblinear,Bubeck15,Zhu17,Acceleration}. 

Since non-linear SVMs with kernel functions can be interpreted as linear SVMs operating in a high-dimensional feature space,~\citet{RFF} first proposed \textit{random feature maps} for approximating shift-invariant kernels to combine the advantages of linear and non-linear SVM approaches.
Their approach approximates kernels by a randomized feature map from data space into a relatively low-dimensional feature space. 
In this randomized feature space, the kernel function of any two vectors is well approximated by their inner product with high probability.
In other words, the randomized feature map can be seen as a specific dimensionality reduction from the feature space, computed directly from vectors in data space.
This enables the use of fast linear learning algorithms to approximate the performance of non-linear kernel methods, significantly reducing training time while maintaining competitive generalization performance.

Given any two vectors $\bx = \left(x_1, \ldots, x_d \right), \by = \left(y_1, \ldots, y_d \right) \in \mR^d$, we denote their inner product by $\dotxy = \sum_{i=1}^d x_i y_i$. 
For an implicit \textit{feature space} map $\phi: \mathbb{R}^d \mapsto \mathcal{F}$ the inner product between vectors in the feature space $\mathcal{F}$ can be quickly computed as $\left\langle \phi(\bx), \phi(\by) \right\rangle = \kappa(\bx, \by)$ where $\kappa(\cdot)$ is an easily computable kernel function. 
A random feature map $f: \mathbb{R}^d \mapsto \mathbb{R}^D$, where $D$ is an integer parameter, can be used to approximate a kernel $\kappa(\cdot)$ if it satisfies
$$\mathbb{E}\left[\left\langle f(\bx), f(\by) \right\rangle\right] = \left\langle \phi(\bx), \phi(\by) \right\rangle = \kappa(\bx, \by), $$
with good concentration around the expected value.
Using such feature maps we can transform data from the original data space into a $D$-dimensional randomized feature space to efficiently approximate the solutions of learning algorithms on high-dimensional feature spaces.
This line of work enables kernel methods to handle massive datasets on many standard statistical learning tasks, including kernel ridge regression~\citep{Avron17,Avron17ICML}, support vector machines~\citep{RFF,OnlineKernel}, clustering~\citep{SpectralClustering,sDbscan} and dimensionality reduction~\citep{CCA,aKPCA}.

Randomized techniques for kernel approximation fall into two broad categories: data-dependent and data-independent methods. 
A prominent example of data-dependent approaches is  Nystr\"{o}m method~\citep{Nystrom00,Nystrom-RFF,Nystrom16}, which approximates the kernel matrix using a subset of training data.
More recent variants leverage statistical properties such as leverage scores of the kernel matrix to select informative samples for improved approximation~\citep{Avron17ICML,UnifiedRFF,RFF_AAAI20}. 

In contrast, data-independent random feature maps~\citep{RFF} approximate the entire kernel function -- not just the matrix -- using features drawn from fixed distributions. This makes them particularly well-suited for online learning and streaming settings, where access to the full training set may be limited. Understanding the theoretical and practical properties of such random features has been the focus of extensive research over the past decade~\citep{ Sriperumbudur15,Bach17,Rudi17,Sun18,UnifiedRFF}.
These efforts have extended to downstream machine learning applications~\citep{KME,RFFSurvey}, as well as to scaling deep learning architectures and analyzing over-parameterized neural networks ~\citep{NN_RFF19,NTK_RFF,Performers,RFF_NN23}.

{\bf Our contribution.}
This paper investigates algorithmic aspects of data-independent random features, with a focus on polynomial kernels.
Although many modern learning algorithms can be trained in time linear in the number of samples~\citep{Bubeck15}, the cost of computing random features often becomes a bottleneck. 
Specifically, many existing kernel approximation methods require time and space proportional to the product of the data dimension $d$ and the number of random features $D$. 
For instance,~\citet{RFF} maintained $D$ random vectors $\bw_1, \ldots, \bw_D \in \mathbb{R}^d$, requiring 
$\BO{dD}$ time and space to compute 
$D$ features. 
When $D = \BO d$, this leads to quadratic costs, which may exceed the time spent in the actual learning or prediction phase. As a result, the cost of the random mapping itself can dominate the overall runtime of kernel-based methods.

We study a near-linear time random feature mapping for approximating the standard polynomial kernel $\kappa(\bx, \by) = \left(c + \left\langle \bx, \by \right\rangle \right)^p$ for an integer $p \geq 1$ and a real $c \geq 0$.
The polynomial kernel is a widely used example of a \textit{non-stationary} kernel, and it serves as a building block for a broader class of kernels that can be well approximated by polynomial expansions, including Gaussian kernels, general dot product kernels, arc-cosine and sigmoid kernels~\citep{TaylorPolyKernel}.
Unlike shift-invariant kernels, polynomial kernels in $\mRd$ do not admit representations via spherical harmonics, and thus require different techniques for constructing random feature maps~\citep{Schoenberg,Bochner}.

In this paper we present \textit{Tensor Sketch}, a scalable random feature map for polynomial kernels that combines the explicit feature mapping via tensor products with an efficient sketching technique to approximate these products.
Given a dataset of $n$ points, Tensor Sketch computes 
$D$-dimensional feature embeddings in time $\BO{n(d + D\log{D})}$ and requires $\BO{1}$ extra space to store the randomness of the sketch.
The core technical insight is a connection between tensor products and the fast convolution structure of Count Sketch~\citep{CountSketch, Pagh13}, which allows for substantial reductions in both computational and memory costs. Empirical evaluations demonstrate that Tensor Sketch achieves high approximation accuracy while outperforming prior methods by orders of magnitude in runtime on large-scale datasets. A preliminary version of this work appeared in the 19th ACM SIGKDD International Conference on Knowledge Discovery and Data Mining~\citep{TS}.

{\bf Later developments.}
Since its initial publication, Tensor Sketch has seen widespread adoption across multiple domains. It has been used to accelerate machine learning algorithms~\citep{cui2017kernel, Dai17} and standard statistical learning tasks~\citep{Avron14, Wang15,diao2018sketching, Draief18}, as well as to improve computational efficiency in various computer vision applications~\citep{Gao16, Fukui16}. Notably, Tensor Sketch has been integrated into widely used libraries such as scikit-learn~\citep{scikit}, where it is available as \textsc{PolynomialCountSketch}.~\footnote{\url{https://scikit-learn.org/stable/modules/generated/sklearn.kernel_approximation.PolynomialCountSketch.html}}

In this extended version, we introduce several important additions and refinements: 
\begin{itemize}
	\item A survey of state-of-the-art techniques for approximating polynomial kernels, along with an overview of recent applications that utilize Tensor Sketch as a core component.
    \item A revised theoretical analysis of the approximation error.
\footnote{Lemma 6 in the conference version contained an error; we correct the dependence on 
$p$.}
In particular, our main result (Theorem~\ref{thm:main}) establishes a variance bound on the inner product of two Tensor Sketch vectors, providing formal guarantees on the reliability and accuracy of the approximation.
\end{itemize}


The organization of the paper is as follows. 
In Section~\ref{sec:related}, we briefly review related work and recent applications using Tensor Sketch as a core algorithmic component. 
Section~\ref{sec:background} describes background and preliminaries. 
Tensor Sketch is presented and analyzed in Section%
~\ref{sec:tensorsketch}. 

\section{Related Work}\label{sec:related}

{\bf Decomposition methods.}
Traditional techniques for training non-linear SVMs on large-scale datasets include decomposition methods~\citep{Osuna97, Scholkopf02, Libsvm}. These algorithms partition the training data into a working set and a fixed set, and iteratively optimize the dual objective with respect to the working set while holding the fixed set constant. In effect, they perform coordinate ascent over subsets of the dual variables, updating only a small number of coefficients per iteration until the Karush-Kuhn-Tucker (KKT) conditions are satisfied within a specified tolerance. While decomposition methods avoid the memory overhead of constructing the full kernel matrix, they rely on repeated numerical optimization steps, which can be computationally expensive for large datasets.

{\bf Data-dependent feature maps.}
To scale kernel methods to large datasets, various techniques have been developed to efficiently approximate the kernel matrix, notably Nystr\"{o}m methods~\citep{Nystrom00,Kumar12,Nystrom16,Avron17ICML,UnifiedRFF,RFF_AAAI20}.
These \textit{data-dependent} approximation methods aim to reduce the time and space complexity while maintaining approximation quality.
Nystr\"{o}m  methods sample a subset of columns (or data points) from the kernel matrix, typically using a distribution informed by the data, and compute a rank-$k$ approximation in time $\BO{nk^2 + k^3}$. Leverage-score-based methods~\citep{UnifiedRFF} proposed a leverage weighted scheme to construct refined random feature maps, achieving kernel approximation in $\BO{nD^2 + D^3}$ time for 
$D$ features.

{\bf Data-independent feature maps.}
Instead of approximating the kernel matrix, \citet{RFF} introduced \textit{Random Fourier Features} (RFF) to approximate shift-invariant kernel functions, including Gaussian, Laplacian, and Cauchy kernels. Since this method approximates the entire kernel function, not just the kernel matrix, it is especially relevant for out-of-sample prediction and online learning scenarios.
The method is based on Bochner’s theorem~\citep{Bochner}, which relates any continuous shift-invariant positive-definite kernel to the Fourier transform of a probability distribution. RFF uses Monte Carlo sampling to approximate the associated integral, mapping input data to a randomized low-dimensional feature space where inner products approximate the kernel values.

This seminal work has opened a new direction for kernel approximation techniques and has been widely adopted in machine learning applications. Since its introduction, hundreds of subsequent research papers have built on RFF, advancing both its algorithmic performance and theoretical guarantees. On the algorithmic side, researchers have proposed more efficient constructions of RFF and studied their variance reduction properties. On the theoretical side, work has focused on the approximation error of the kernel matrix, the expected risk of learning algorithms, and their generalization properties when using RFF. We refer readers to the survey by~\citet{RFFSurvey} and references therein for a comprehensive overview.

\begin{table} [t]\label{Tab:Summary}
	\centering	
	\caption{A summary of recent works on approximating polynomial kernels $\left(c + \left\langle \bx, \by \right\rangle \right) ^p$ where $c \geq 0, p = \BO{1}$, their construction time and extra memory to store the randomness.}
	\begin{tabular}{|c|c|c|c|} 
		\hline
		Technique & Reference & Time & Memory \\
		\hline
		\multirow{2}{*}{Sketching} & ~\citet{Kar12}	& \multirow{2}{*}{$\BO{dD}$} & \multirow{2}{*}{$\BO{D}$} \\
        & ~\citet{Meister19} & &	 \\
        \hline
	Sketching		& ~\citet{Hamid14}	& $\BO{dD + D\log{D}}$ & $\BO{D}$ \\
        \hline
		\multirow{2}{*}{Monte Carlo method}		& ~\citet{Pennington15}	& \multirow{2}{*}{$\BO{dD}$} & \multirow{2}{*}{$\BO{dD}$} \\
        & ~\citet{Krein} & &
         \\ 
         \hline
         
	 Sketching & \textbf{Tensor Sketch}	& $\pmb{\BO{D\log{D}}}$ & \bm $\pmb{\BO{1}}$ \\
		\hline
	\end{tabular}
	\label{tab:summary}
\end{table}

{\bf Polynomial kernels.}
Our work focuses on developing data-independent random feature maps for polynomial kernels, a class of non-stationary kernels of the form
$$\kappa(\bx, \by) = \left(c + \left\langle \bx, \by \right\rangle \right)^p \, ,$$ for an integer $p \geq 1$ and a constant $c \geq 0$.
Unlike shift-invariant kernels, polynomial kernels defined over 
$\mRd$ do not admit spectral representations via spherical harmonics, making standard random feature techniques inapplicable. As a result, constructing efficient random feature maps for polynomial kernels typically requires tools from sketching and dimensionality reduction in linear algebra~\citep{Woodruff_Book}.
We provide a detailed review of recent methods for approximating polynomial kernels, highlighting the advantages and limitations of each approach. Table~\ref{tab:summary} summarizes these methods along with their time and space complexities.

~\citet{Kar12} proposed random feature maps based on the Maclaurin series expansion of \textit{inhomogeneous} polynomial kernels, i.e., $\kappa(\bx, \by) = \left(c + \left\langle \bx, \by \right\rangle \right) ^p$.
Their method approximates the \textit{homogeneous} polynomial kernels $\kappa(\bx, \by) = \left\langle \bx, \by \right\rangle ^p$ by $$\left( \prod_{i=1}^{p}{\left\langle \bw_i, \bx \right\rangle} \right) \left( \prod_{i=1}^{p}{\left\langle \bw_i, \by \right\rangle} \right) \, , $$ where each $\bw_i \in \left\{ -1, 1 \right\}^d$ is a Rademacher vector.
To construct one random feature for inhomogeneous kernels, the method randomly selects a degree-$t$ term $\left\langle \bx, \by \right\rangle ^t$ with probability 
$1/2^{t+1}$, corresponding to the 
$t$-th order term in the Taylor expansion, and uses 
$t$ Rademacher vectors to construct the feature. The total runtime is 
$\BO{dD}$, where 
$D$ is the number of random features.
This approach was later referred to as \textit{Tensorized Random Projection} by \citet{Meister19}, who also provided improved high-probability error bounds (see Theorem 2.2).

~\citet{Hamid14} proposed CRAFTMaps which build on the Maclaurin-based method by incorporating fast Johnson-Lindenstrauss transforms~\citep{FJLT,Tropp11}.
It first constructs $D' = \BO{D}$ random features using the Maclaurin-based method~\citep{Kar12}, and then applies Subsampled randomized Hadamard (SRHT) to reduce the dimensionality from
$D'$ to $D$ features in $\BO{D \log {D}}$ time.
The overall runtime is 
$\BO{dD+D\log{D}}$, where the first term accounts for feature construction and the second for the projection.

\citet{Pennington15} proposed random feature maps for a specific class of polynomial kernels defined on the unit sphere. 
It employs mixtures of Gaussian distributions to generate 
$D$ spherical random Fourier features, which approximate the non positive-definite kernel $$\kappa(\bx, \by) = \left( 1 - \frac{1}{c} + \frac{1}{c}\left\langle \bx, \by \right\rangle \right)^p \, ,$$ where $c \geq 2$ and $\| \bx \|_2 = \| \by \|_2 = 1$.
Like standard random Fourier methods, this construction requires 
$\BO{dD}$
 time and storage. Although this method introduces bias in the kernel approximation, \citet{Krein} later demonstrated how to construct unbiased estimators for a broader class of non positive-definite kernels, also with 
$\BO{dD}$ computational complexity.

Building on our work on Tensor Sketch,~\citet{ahle2020oblivious} introduced a method that improves the dependence on the polynomial degree 
$p$. Specifically, their approach requires an embedding dimension 
$D$ that scales polynomially with 
$p$ to achieve a target error, in contrast to the exponential dependence in the original Tensor Sketch. 
While Tensor Sketch can be followed by a separate dimension reduction step, ~\citet{ahle2020oblivious} show that it is more efficient to interleave sketching and projection through a tree-structured composition of tensoring operations. This recursive strategy maintains computational efficiency and allows the overall approximation error to be tightly bounded.

\section{Background and Preliminaries}\label{sec:background}

This section introduces tensor powers, the explicit feature map for polynomial kernels, which expands input vectors into a space of exponentially growing dimension. To address the computational challenges arising from this expansion, we review key sketching techniques including AMS Sketches and Count Sketches, which serve as efficient random projection methods for approximating inner products in the high-dimensional kernel space.

\subsection{Notation}

For an integer $d$ we use $[d]$ to denote the set $\{1,\dots,d\}$.
Consider a vector $\bx = \left(x_1, \ldots, x_d \right) \in \mathbb{R}^d$.
For $q > 0$ the $\ell_q$ norm of $\bx$ is defined as
$\| \bx \|_q = \left(\sum_{i=1}^d |x_i|^q \right)^{1/q}$.
%
Given $\by = \left(y_1, \ldots, y_d \right) \in \mathbb{R}^d$, we define
$\left\langle \bx, \by \right\rangle = \sum_{i=1}^d x_i y_i$.
%
%
The 2nd \textit{tensor power} of $\bx$ (i.e. outer product $\bx \otimes \bx$)  denoted by $\bx^{(2)}$, is defined as the vector with entries: 
$$
\bx^{(2)} = \bx \otimes \bx = 
\begin{bmatrix}
x_1x_1 & x_1x_2 & \cdots & x_1x_d \\
x_2x_1 & x_2x_2 & \cdots & x_2x_d \\
\vdots & \vdots & \ddots & \vdots \\
x_dx_1 & x_dx_2 & \cdots & x_dx_d 
\end{bmatrix} \enspace .$$
Though it is depicted in matrix form, we will think about $\bx^{(2)}$ as a vector in $\mathbb{R}^{d^2}$ (with any fixed ordering of the entries).
Generally, given an integer $p > 1$ we consider a $p$-th tensor power~\footnote{$p$-th tensor power of $\bx$: $\bx^{(p)} = \bx \underbrace{\otimes \cdots \otimes}_{\text{$p$ times}} \bx $.} 
$\bx^{(p)}$ indexed by vectors in $[d]^p$, such that:
$$ \bx^{(p)}_{(i_1,\dots,i_p)} = \prod_{j=1}^p x_{i_j} \enspace. $$
%

\subsection{Tensor Powers as Feature Maps}

\citet[Proposition~2.1]{Scholkopf02} justifies that tensor power is an explicit feature map for the homogeneous polynomial kernel.
\begin{lemma}
	\label{lem:TP1}
	Given any pair of vectors $\bx, \by \in \mRd$ and an integer $p \geq 1$, we have:
	\begin{displaymath}
	\begin{aligned}
	\left\langle \bx^{(p)}, \by^{(p)} \right\rangle &= \left\langle \bx, \by \right\rangle^p \enspace.\\
	\end{aligned}
	\end{displaymath}
\end{lemma}
%


Since the tensor power map requires $d^p$ dimensions, it is not feasible to explicitly compute data coordinates in kernel space for high-dimensional vectors.

\subsection{AMS Sketches}

Most sketching approaches require $k$-wise independent hash families, defined as follows.

\begin{definition}
	A family $\mathcal{H}$ of functions $f: [u] \rightarrow [r]$ is $k$-wise independent if for every set $I\subseteq [u]$ of $|I|=k$ elements, and a random $f\in\mathcal{H}$ the vector of hash values $(f(i))_{i\in I}$ is uniform in $[r]^k$.
\end{definition}

\citet{AMS} described and analyzed a sketching approach, referred to as the \textit{AMS Sketch}, to estimate the second frequency moment of a high-dimensional vector.
AMS Sketch samples random functions from a 4-wise independent family. 
Such samples can be generated efficiently by storing $\BO{1}$ random integers, and the sampled functions can be evaluated in $\BO{1}$ time~\citep{Carter79}.

\begin{definition}	
	Given $s: [d] \mapsto \left\{-1, 1\right\}$ sampled from a 4-wise independent family, an \textit{AMS Sketch} $Z(\bx)$ of a vector $\bx = \left(x_1, \ldots, x_d \right) \in \mathbb{R}^d$ is computed as $Z(\bx) = Z_s(\bx) = \sum_{i=1}^{d} s(i)\, x_i$.
\end{definition}

We use the subscript $s$ to describe the corresponding hash function $s$ of the AMS Sketch.
We will skip it when the context is clear.
\citet{AMS} analyze the bias and variance of AMS Sketch.
\begin{lemma}(\citet[Theorem~2.2]{AMS})\label{lem:AMS}
	Consider an AMS sketch $Z(\bx)$ for $\bx \in \mathbb{R}^d$. We have $\E{Z(\bx)^2} = \|\bx\|_2^2$ and $\Var{Z(\bx)^ 2} \leq 2\,\|\bx\|_2^4$.
\end{lemma}

%

%
%

\subsection{Count Sketch}

\citet{CountSketch} described and analyzed a sketching approach, called \textit{Count Sketch}, to estimate the frequency of items in a data stream. 
\begin{definition}\label{def:count-sketch}
Given $h: [d] \mapsto [D]$ sampled from a 2-wise independent family, and $s: [d] \mapsto \left\{-1, 1\right\}$ sampled from a 4-wise independent family, a \textit{Count Sketch} $\bCx$ of a vector $\bx = \left(x_1, \ldots, x_d \right) \in \mathbb{R}^d$ is computed as $\bCx = \left((Cx)_1, \ldots, (Cx)_D \right) \in \mathbb{R}^D$ where $(Cx)_k = \sum_{i: h(i)=k}{s(i)\, x_i}.$
\end{definition}

The Count Sketch definition above is a slight relaxation of the original definition proposed by~\citet{CountSketch} where we require the hash function $s$ to be sampled from a \textit{4-wise} independent family.
\citet{Weinberger09} introduced a variant of Count Sketch as a feature hashing method for large-scale multi-task learning, leveraging its ability to approximately preserve inner products.
Inspired by this property, our work applies Count Sketch in a similar fashion, but instead of operating in the input space, we use it to project implicitly into the feature space of polynomial kernels, without explicitly constructing the high-dimensional feature vectors.

Count Sketch can be seen as a \textit{random projection} technique because it computes linear projections of $\bx$ with random vectors implicitly defined by hash functions $h$ and $s$.
The following lemma provides the bias and variance on inner product of Count Sketches.

\begin{lemma}\label{lem:CS}
Given vectors $\bx, \by \in \mathbb{R}^d$, denote by $\bCx, \bCy \in \mathbb{R}^D$ the respective Count Sketches of $\bx, \by$ based on the hash functions $h, s$. 
Then we have:
\begin{displaymath}
\begin{aligned}
\E{\left\langle \bCx, \bCy \right\rangle} &= \left\langle \bx, \by \right\rangle ,\\
\Var{\left\langle \bCx, \bCy \right\rangle}  &= \frac{1}{D} \left( \sum_{i \neq j}{x_i^2 y_j^2} + \sum_{i \neq j}{x_iy_ix_jy_j}  \right) \leq \frac{2}{D} \, \|\bx\|_2^2 \, \|\by\|_2^2 \enspace.
%
\end{aligned}
\end{displaymath}
\end{lemma}

\begin{proof}
Define the indicator variable $\xi_{ij} = \mathbb{I}[h(i) = h(j)]$. Then we can write:
\[
\langle \bCx, \bCy \rangle = \sum_{i,j} x_i y_j s(i)s(j) \xi_{ij} = \langle \bx, \by \rangle + \sum_{i \ne j} x_i y_j s(i) s(j) \xi_{ij} \, .
\]
Taking expectation and using independence:
\[
\mathbb{E}[\langle \bCx, \bCy \rangle] = \langle \bx, \by \rangle + \sum_{i \ne j} x_i y_j \, \mathbb{E}[s(i)] \, \mathbb{E}[s(j)] \, \mathbb{E}[\xi_{ij}] = \langle \bx, \by \rangle,
\]
since $\mathbb{E}[s(i)] = 0$ and $\mathbb{E}[\xi_{ij}] = \Pr{h(i) = h(j)} = \frac{1}{D}$.
For the variance, compute $\mathbb{E}[\langle \bCx, \bCy \rangle^2]$:
\begin{align*}
\langle \bCx, \bCy \rangle^2 
&= \left( \langle \bx, \by \rangle + \sum_{i \ne j} x_i y_j s(i) s(j) \xi_{ij} \right)^2 \\
&= \langle \bx, \by \rangle^2 + 2 \langle \bx, \by \rangle \sum_{i \ne j} x_i y_j s(i) s(j) \xi_{ij} + \left( \sum_{i \ne j} x_i y_j s(i) s(j) \xi_{ij} \right)^2.
\end{align*}

The middle term has zero expectation due to independence and zero-mean of $s(i)$. Expanding the last term, and keeping only terms whose expectation survives due to 4-wise independence of $s$, we get:
\[
\mathbb{E}\left[\left( \sum_{i \ne j} x_i y_j s(i) s(j) \xi_{ij} \right)^2\right] = \frac{1}{D} \left( \sum_{i \ne j} x_i^2 y_j^2 + \sum_{i \ne j} x_i y_i x_j y_j \right).
\]
%
Hence, we have
\[
\mathrm{Var}[\langle \bCx, \bCy \rangle] = \mathbb{E}[\langle \bCx, \bCy \rangle^2] - \langle \bx, \by \rangle^2 = \frac{1}{D} \left( \sum_{i \ne j} x_i^2 y_j^2 + \sum_{i \ne j} x_i y_i x_j y_j \right).
\]

Finally, using Cauchy-Schwarz and basic norm identities:
\[
\sum_{i \ne j} x_i^2 y_j^2 \le \|\bx\|_2^2 \|\by\|_2^2, \quad 
\sum_{i \ne j} x_i y_i x_j y_j \le \|\bx\|_2^2 \|\by\|_2^2 \, ,
\]
we conclude:
\[
\mathrm{Var}[\langle \bCx, \bCy \rangle] \le \frac{2}{D} \|\bx\|_2^2 \|\by\|_2^2 \enspace. 
\]
\end{proof}





\subsection{AMS Sketch on the Tensor Domain}~\label{sec:ams}

Similar to the Maclaurin-based approach of \citet{Kar12}, AMS Sketches can be employed as random features to approximate polynomial kernels. Below, we focus on the construction of a single random feature; computing 
$D$ features is straightforward by independently repeating this process 
$D$ times.
\citet{Indyk08}, followed by \citet{Braverman10}, analyzed the use of products of AMS Sketches with independently chosen hash functions, and established the following result:

\begin{lemma}~\cite[Lemma 4.1]{Braverman10}\label{lem:AMS1}
	Consider $\bx \in \mathbb{R}^{d}$, an integer $p > 1$, and $p$ independently sampled functions $s_1, \ldots, s_p: [d] \mapsto \{-1, 1\}$ from a 4-wise independent family.
	Define $Z = \prod_{j=1}^p Z_{s_j}(\bx)$. Then $\E{Z^2} = \| \bx^{(p)} \|_2^{2} = \| \bx \|_2^{2p}$ and $\Var{Z^2} \leq  (3^p - 1)\|\bx^{(p)}\|_2^{4} = (3^p - 1)\|\bx\|_2^{4p}$. 
\end{lemma}

Since $Z$ is defined as the product of independent AMS sketches, its expectation $\mathbb{E}[Z^2]$ can be directly derived from Lemma~\ref{lem:AMS}. However, analyzing the variance $\mathrm{Var}(Z^2)$ is significantly more challenging and constitutes the main contribution of \citet{Braverman10}.

To proceed, consider defining a composite hash function $S(x_{i_1}, \ldots, x_{i_p}) = \prod_{j=1}^p s_j(x_{i_j})$, which maps a coordinate $x_{i_1}x_{i_2} \ldots x_{i_p}$ of $\bx^{(p)} \in [d]^p$ to $\{-1, 1\}$. Under this definition, $Z$ can be interpreted as an AMS sketch applied to the $p$-th order tensor product $\bx^{(p)}$:
$$Z = Z_{S}(\bx^{(p)}) = \prod_{j=1}^p Z_{s_j}(\bx) \, .$$
However, the hash function $S : [d]^p \rightarrow \{-1, 1\}$ is not drawn from a 4-wise independent family, and thus standard AMS analysis techniques are insufficient to bound the variance.
\citet{Braverman10} address this by analyzing the combinatorial structure of correlations between $S(\bu)$ and $S(\bv)$ for distinct index vectors $\bu, \bv \in [d]^p$. 

Based on this approach, it is natural to generalize Lemma~\ref{lem:AMS1} to handle a pair of tensorized vectors $\bx^{(p)}$ and $\by^{(p)}$, as follows:
\begin{lemma}\label{lem:AMS2}
	Consider $\bx, \by \in \mathbb{R}^{d}$, an integer $p > 1$, and $p$ independently sampled functions $s_1, \ldots, s_p: [d] \mapsto \{-1, 1\}$ from a 4-wise independent family.
	Define $Z = \prod_{j=1}^p Z_{s_j}(\bx)\,Z_{s_j}(\by)$. 
	Then $\E{Z} = \dotxy^{p}$ and $\Var{Z} \leq (3^p - 1) \, \|\bx\|_2^{2p} \, \|\by\|_2^{2p}$.
	%
\end{lemma}
%


\begin{proof}
	Following the approach of~\cite{Braverman10}, we compute the expectation and variance of \(Z\).
	First, we consider the expectation, and for each \(j\), we note that
	\[
	\mathbb{E}\bigl[Z_{s_j}(\bx) \, Z_{s_j}(\by)\bigr] = \mathbb{E}\left[\left(\sum_{i=1}^d x_i\, s_j(i)\right) \left(\sum_{k=1}^d y_k\, s_j(k)\right)\right] = \langle \bx, \by \rangle,
	\]
	where the last equality follows from the 4-wise independence and unbiasedness of the \(s_j\). Since the functions \(s_j\) are independent for different \(j\), we have
	\[
	\mathbb{E}[Z] = \prod_{j=1}^p \mathbb{E}\bigl[Z_{s_j}(\bx)\, Z_{s_j}(\by)\bigr] = \langle \bx, \by \rangle^p \, .
	\]

	Next, we bound the variance \(\operatorname{Var}(Z) = \mathbb{E}[Z^2] - (\mathbb{E}[Z])^2\). Because the hash functions are independent across different \(j\), we may write
	\begin{equation}\label{eq:secondmoment}
		\mathbb{E}[Z^2] = \prod_{j=1}^p \mathbb{E}\Bigl[ \bigl(Z_{s_j}(\bx)\, Z_{s_j}(\by)\bigr)^2 \Bigr] \, .
	\end{equation}
	For each \(j\), expanding the square gives
	\begin{align*}
		\mathbb{E}\Bigl[\bigl(Z_{s_j}(\bx)\, Z_{s_j}(\by)\bigr)^2\Bigr] 
		& = \mathbb{E}\left[ \left( \sum_{i=1}^d x_i\, s_j(i) \right)^2 \left( \sum_{k=1}^d y_k\, s_j(k) \right)^2 \right] \\
		& = \sum_{i=1}^d \sum_{k=1}^d \sum_{i'=1}^d \sum_{k'=1}^d x_i\, y_k\, x_{i'}\, y_{k'}\, \mathbb{E}\left[ s_j(i)\, s_j(k)\, s_j(i')\, s_j(k') \right] \, .
	\end{align*}
	Observing that \(\mathbb{E}\left[s_j(i)s_j(k)s_j(i')s_j(k')\right]\) is nonzero only when the indices form pairs (including the possibility that all four are identical), we have
	\[
	\mathbb{E}\left[ s_j(i)s_j(k)s_j(i')s_j(k') \right] = 
	\begin{cases}
		1 & \text{if } i = k = i' = k', \\
		1 & \text{if } i = k \neq i' = k', \\
		1 & \text{if } i = i' \neq k = k', \\
		1 & \text{if } i = k' \neq k = i', \\
		0 & \text{otherwise}.
	\end{cases}
	\]
	
	The contribution from terms with \(i = k = i' = k'\) is $\sum_{i=1}^d x_i^2 y_i^2$.
	Terms with \(i = k \neq i' = k'\) contribute
	\[
	\sum_{i\neq i'} x_i\, y_i\, x_{i'}\, y_{i'} = \left(\sum_{i=1}^d x_i\, y_i\right)^2 - \sum_{i=1}^d x_i^2 y_i^2 = \langle \bx, \by \rangle^2 - \sum_{i=1}^d x_i^2 y_i^2 \, .
	\]
	The case \(i = k' \neq k = i'\) is symmetric and yields the same contribution. Finally, for \(i = i' \neq k = k'\) we obtain
	\[
	\sum_{i\neq k} x_i^2\, y_k^2 = \|\bx\|_2^2\, \|\by\|_2^2 - \sum_{i=1}^d x_i^2 y_i^2 \, .
	\]
	Thus, summing these contributions, we have
	\begin{align*}
    \mathbb{E}\Bigl[\bigl(Z_{s_j}(\bx)\, Z_{s_j}(\by)\bigr)^2\Bigr] &= \sum_{i=1}^d x_i^2 y_i^2 + 2\Bigl(\langle \bx, \by \rangle^2 - \sum_{i=1}^d x_i^2 y_i^2\Bigr) + \Bigl(\|\bx\|_2^2\, \|\by\|_2^2 - \sum_{i=1}^d x_i^2 y_i^2\Bigr) \\
    &= 2\langle \bx, \by \rangle^2 + \|\bx\|_2^2\, \|\by\|_2^2 - 2\sum_{i=1}^d x_i^2 y_i^2 \, .
	\end{align*}
	Using the Cauchy--Schwarz inequality, \(\langle \bx, \by \rangle^2 \leq \|\bx\|_2^2\, \|\by\|_2^2\), and noting that \(\sum_{i=1}^d x_i^2 y_i^2 \geq 0\), it follows that
	\[
	\mathbb{E}\Bigl[\bigl(Z_{s_j}(\bx)\, Z_{s_j}(\by)\bigr)^2\Bigr] \leq 3\, \|\bx\|_2^2\, \|\by\|_2^2 \, .
	\]
	Substituting this bound into \eqref{eq:secondmoment} yields
	\[
	\mathbb{E}[Z^2] \leq \Bigl( 3\, \|\bx\|_2^2\, \|\by\|_2^2 \Bigr)^p,
	\]
	which completes the proof since
	\[
	\operatorname{Var}(Z) = \mathbb{E}[Z^2] - \langle \bx, \by \rangle^{2p} \leq 3^p\, \|\bx\|_2^{2p}\, \|\by\|_2^{2p} \, .
	\]
\end{proof}

Lemma~\ref{lem:AMS2} shows that the AMS sketch can be interpreted as a variant of the Maclaurin-based approach of \citet{Kar12}, where the Rademacher vectors $\mathbf{w}_i$ are replaced by hash functions $s$ drawn from a 4-wise independent family. This substitution enables a provable variance bound for the resulting random feature maps.

\section{Tensor Sketch}\label{sec:tensorsketch}

We now describe how to compute the Count Sketch of a tensor product $\bx^{(p)}$, which serves as a random feature map for approximating polynomial kernels $\kappa(\bx, \by) = \dotxy^p$, for an integer $p > 0$. 
For the kernel  $\kappa(\bx, \by) = \left( c + \dotxy \right)^p$,  we can avoid the constant $c > 0$ by adding an extra dimension of value $\sqrt{c}$ to all data vectors. 
For a fixed polynomial degree $p > 1$, our method maps an input vector $\bx \in \mathbb{R}^d$ to a low-dimensional feature vector in $\mathbb{R}^D$ in time $\mathcal{O}(d + D \log D)$.
This runtime represents a significant improvement over the AMS Sketch-based approach described in Section~\ref{sec:ams}, both in terms of computational efficiency and practical scalability.

\subsection{Convolution of Count Sketches}

\citet{Pagh13} introduced a fast algorithm for computing the Count Sketch of the outer product of two vectors. As our method builds on this technique, we briefly review it here.
Rather than explicitly forming the outer product, the key idea is to first compute the Count Sketches of the input vectors and then derive the sketch of their outer product directly from these compressed representations. As we will show, this process reduces to polynomial multiplication, which can be efficiently implemented using the Fast Fourier Transform (FFT).~\footnote{
For background on FFT, see for example \citet[Section~5.6]{Kleinberg05}.}
This yields an algorithm to compute the Count Sketch of an outer product in time near-linear in the size of the sketches.

More precisely, given a vector $\bx \in \mathbb{R}^d$, we denote by $\bCxone, \bCxtwo \in \mathbb{R}^D$ two different Count Sketches using hash functions $h_1, h_2 : [d] \mapsto [D]$ and $s_1, s_2: [d] \mapsto \left\{ -1, 1 \right\}$, all independently sampled from \textit{2-wise} independent families.
The aim is to compute a Count Sketch of the outer product $\bx \otimes \bx \in \mathbb{R}^{d^2}$, denoted $\bCx^{(2)} \in \mathbb{R}^D$, from $\bCxone$ and $\bCxtwo$. 
We define the Count Sketch $\bCx^{(2)}$ in terms of the hash functions $H : [d^2] \mapsto [D]$ and $S: [d^2] \mapsto \left\{ -1, 1 \right\}$ derived from the functions $h_1, h_2, s_1, s_2$, where:
%
\begin{align}\label{eq:Hash}
	H(i_1, i_2) = \left( h_1(i_1) + h_2(i_2) \right) \text{ mod } D \text{ and } S(i_1, i_2) = s_1(i_1)s_2(i_2) \enspace .
\end{align}

It is well known that since $h_1$, $h_2$, $s_1$, $s_2$ are sampled from 2-wise independent families, $H$ and $S$ are also sampled from 2-wise independent families~\citep{Patrascu11}. 
Na\"ively computing $\bCx^{(2)}$ would require $\BO{d^2}$ time.
However, by thinking of Count Sketches as polynomials, we are able to exploit FFT to fast compute $\bCx^{(2)}$ given hash functions $H$ and $S$ defined in Equation~\ref{eq:Hash}.

In particular, we represent Count Sketches $\bCxone, \bCxtwo \in \mathbb{R}^D$ as polynomials of degree $D-1$
where each entry $(\bCxone)_k$ or $(\bCxtwo)_k$ is the coefficient of $\omega^k$ in the polynomial:
%
$$P_\bx(\omega) = \sum_{i=1}^{d}{s_1(i)x_i \, \omega^{h_1(i)}} \text{ and } Q_\bx(\omega) = \sum_{i=1}^{d}{s_2(i)x_i \, \omega^{h_2(i)}} \, .$$
In other words, the $k$-th features of $\bCxone$ and $\bCxtwo$ corresponding to the coefficients of the term $\omega^k$ are $\sum_{i, h_1(i) = k} s_1(i)x_i$ and $\sum_{i, h_2(i) = k} s_2(i)x_i$, respectively.
We will derive the polynomial multiplication $P_\bx(\omega) \times Q_\bx(\omega)$ as follows.

\begin{align*}
	&P_\bx(\omega) \times Q_\bx(\omega) = \left( \sum_{i=1}^{d}{s_1(i)x_i \, \omega^{h_1(i)}} \right) \left( \sum_{i=1}^{d}{s_2(i)x_i \, \omega^{h_2(i)}} \right) \\
	&= \sum_{i_1,i_2=1}^{d}{s_1(i_1)s_2(i_2)x_{i_1}x_{i_2} \, \omega^{h_1(i_1) + h_2(i_2)}} 
	= \sum_{i_1,i_2=1}^{d}{S(i_1, i_2)x_{i_1}x_{i_2} \, \omega^{h_1(i_1) + h_2(i_2)}} \, . 
\end{align*} 

We note that the polynomial $P_\bx(\omega) \times Q_\bx(\omega)$ has degree $2D - 2$ since $h_1, h_2: [d] \mapsto [D]$.
We transform the polynomial $P_\bx(\omega) \times Q_\bx(\omega)$ to the polynomial of degree $D - 1$ by casting coefficients of the term $\omega^k$ as the coefficients of the term $\omega^{k \text{ mod } D}$ where $0 \leq k \leq 2D - 2$.
We denote by $P_{\bx^{(2)}}(\omega)$ the transformation polynomial of $(D - 1)$-degree from $P_\bx(\omega) \times Q_\bx(\omega)$.
It is clear that $P_{\bx^{(2)}}(\omega)$ is polynomial representation of the Count Sketch of the $\bCx^{(2)} \in \mRD$ of $\bx^{(2)}$ using $H$ and $S$ as
\begin{align*}
	P_{\bx^{(2)}}(\omega) &= \sum_{i_1,i_2=1}^{d}{S(i_1, i_2)x_{i_1}x_{i_2} \, \omega^{\left( h_1(i_1) + h_2(i_2) \right) \text{ mod } D}} = \sum_{i_1,i_2=1}^{d}{S(i_1, i_2)x_{i_1}x_{i_2} \, \omega^{H(i_1, i_2)}} \, .
\end{align*} 

The $(D-1)$-degree polynomial $P_{\bx^{(2)}}(\omega)$ derived from the polynomial multiplication $P_\bx(\omega) \times Q_\bx(\omega)$ can be computed in time $\BO{D\log{D}}$ using FFT and its inverse FFT$^{-1}$: 
%
%
%
$$P_{\bx^{(2)}}(\omega) = \text{FFT}^{-1} \left( \text{FFT}(P_\bx(\omega)) \circ \text{FFT}(Q_\bx(\omega)) \right) \, ,$$
where $\circ$ is the component-wise product operator defined by $({\bf a} \circ {\bf b})_i = a_i b_i$. 
In other words, the Count Sketch $\bCx^{(2)}$ of $\bx \otimes \bx$ can be efficiently computed by Count Sketches $\bCxone$ and $\bCxtwo$ in $\BO{d + D\log{D}}$ time.
The first term comes from constructing the count sketches and the latter comes from running FFT three times.

\subsection{Tensor Sketch}

We now extend the previous method to compute Tensor Sketch of a $p$-th tensor power $\bxp$.
This is achieved by convolving $p$ independent Count Sketches, each constructed using hash functions $h_1, \ldots, h_p : [d] \to [D]$ and $s_1, \ldots, s_p : [d] \to \{-1, 1\}$.
The hash functions $h_i$ are drawn from 2-wise independent families, but the sign functions $s_i$ are drawn from \textit{4-wise} independent families to ensure variance bounds in the resulting sketch.
This construction yields a Count Sketch for $\bx^{(p)}$, referred to as the $p$th-order \textit{Tensor Sketch}, defined by the following composite hash functions:
\begin{displaymath}
\begin{aligned}
H(i_1, \ldots , i_p) = \left( \sum_{j=1}^{p}{h_j(i_j)} \right) \text{ mod } D \text{ and } 
S(i_1, \ldots , i_p) = \prod_{j=1}^{p}{s_j(i_j)} \, .
\end{aligned}
\end{displaymath}

We leverage the efficient computation of Count Sketches over tensor domains to develop a fast algorithm for approximating the homogeneous polynomial kernel $\kappa(\bx, \by) = \langle \bx, \by \rangle^p$, where $p$ is a positive integer.
For each input vector $\bx \in \mathbb{R}^d$, Tensor Sketch computes a Count Sketch of the $p$-fold tensor product $\bx^{(p)}$, producing a $D$-dimensional random feature map in $\mathbb{R}^D$ that approximates the kernel.
The full procedure is outlined in Algorithm~\ref{algorithm:tensorsketch}, which illustrates how Tensor Sketch efficiently maps input vectors to the lower-dimensional kernel feature space.
\begin{algorithm}[t]
\caption{Tensor Sketch($\bx, p, D$)}
\label{alg1} 									
\begin{algorithmic} [1]
\REQUIRE { Vector $\bx \in \mRd$, integer $D$, $p > 1$, $h_1, \ldots, h_p: [d] \mapsto [D]$ independently sampled from a 2-wise independent family, $s_1, \ldots, s_p: [d] \mapsto \{-1, 1\}$ independently sampled from a 4-wise independent family. }
\ENSURE { Return a feature vector $f(\bx)\in\mathbb{R}^D$ such that $\E{\left\langle f(\bx),f(\by) \right\rangle} = \kappa(\bx, \by) = \left\langle \bx, \by \right\rangle^p$  }
%
	\STATE {For $i = 1, \ldots , p$ create Count Sketch $\bf{C_i x}$ using hash functions $h_i, s_i$ }
	\STATE {For $i = 1, \ldots , p$ let $\bf{\widehat{C}_i x} \gets \text{FFT}(\bf{C_i x}) $ }
	\STATE {Let $\bf{\widehat{C}x} \gets \bf{\widehat C_1x} \circ \ldots \circ \bf{\widehat C_px}$} \hfill (component-wise multiplication)
    \STATE {Let $f({\bf x}) \gets \text{FFT}^{-1}(\bf{\widehat{C}x})$}
%
%
\RETURN {$f({\bx})$}
\normalsize
\end{algorithmic}\label{algorithm:tensorsketch}
\end{algorithm}

We maintain $2p$ independent hash functions: $h_1, \ldots, h_p$, which are 2-wise independent, and $s_1, \ldots, s_p$, which are 4-wise independent.
For each input vector $\bx$, we construct $p$ independent Count Sketches of size $D$ using these hash functions (Line 1 of Algorithm~\ref{algorithm:tensorsketch}).
The final sketch of the $p$-fold tensor product $\bx^{(p)}$ is then computed via fast polynomial multiplication using the Fast Fourier Transform (FFT).
This procedure yields a random feature map $f$ that serves as an unbiased estimator of the homogeneous polynomial kernel for any pair of input vectors.

We now analyze the computational and space complexity of Tensor Sketch.
Since the degree $p$ is fixed and typically small, only $\mathcal{O}(1)$ space is required to store the $2p$ hash functions~\citep{Carter79, Thorup12}.
For each vector, computing the sketch of its $p$-fold tensor product takes $\mathcal{O}(d + D \log D)$ time, due to the use of FFT.
Hence, given $n$ data points, the total runtime of Tensor Sketch is $\mathcal{O}(n(d + D \log D))$.

To improve the accuracy of kernel approximation, $D$ is often chosen to be $\mathcal{O}(d)$, resulting in an overall time complexity of $\mathcal{O}(nd \log d)$.
This is significantly faster than earlier approaches such as \citet{Kar12}, \citet{Hamid14}, and \citet{Pennington15}, which typically require $\mathcal{O}(nd^2)$ time.
Furthermore, Tensor Sketch only requires $\mathcal{O}(1)$ additional memory to store the hash functions, while previous methods demand $\mathcal{O}(d)$ space~\citep{Kar12, Hamid14, Meister19} or even $\mathcal{O}(d^2)$~\citep{Pennington15, Krein}.

\subsection{Error Analysis}\label{sec:error}

In this section, we analyze the estimation accuracy of the polynomial kernel $\kappa(\bx, \by) = \langle \bx, \by \rangle^p$ for a positive integer $p$, focusing on the concentration behavior of the estimator produced by Tensor Sketch.
We derive bounds on the variance of the estimator as a function of the number of random features $D$, and show that the estimate concentrates tightly around its expected value.
The following theorem establishes both the unbiasedness and a variance bound for the approximation provided by Tensor Sketch.

\begin{theorem}\label{thm:main}
Given two vectors $\bx, \by \in \mathbb{R}^d$, we denote by $\bCx^{(p)}, \bCy^{(p)} \in \mathbb{R}^D$ the Count Sketches of $\bx^{(p)}, \by^{(p)} \in \mathbb{R}^{d^p}$ using hash functions $h_1, \ldots , h_p: [d] \mapsto [D]$ and $s_1, \ldots , s_p: [d] \mapsto \left\{ -1, 1 \right\}$ chosen independently from \textit{2-wise} and \textit{4-wise} independent families, respectively.
Then we have
\begin{displaymath}
\begin{aligned}
\E{\left\langle \bCx^{(p)}, \bCy^{(p)} \right\rangle } &= \left\langle \bxp, \byp \right\rangle = \dotxy^{p},\\
\Var{\left\langle \bCx^{(p)}, \bCy^{(p)} 
\right\rangle } &\leq 
\frac{3^p - 1}{D} \|\bx\|_2^{2p}\|\by\|_2^{2p}
\enspace. 
\end{aligned}
\end{displaymath}
\end{theorem}
\begin{proof}
We note that the Tensor Sketches $\bCx^{(p)}, \bCy^{(p)}$ are the Count Sketches of the tensor product $\bX = \bx^{(p)}, \bY = \by^{(p)}$ using the two new hash functions $H : [d]^p \mapsto [D]$ and $S: [d]^p \mapsto \left\{-1, 1 \right\}$ such that:
$$H(i_1, \ldots, i_p) = \left( \sum_{j = 1}^{p}{h_j(i_j)}\right) \text{ mod } D \text{ and } S(i_1, \ldots, i_p) = \prod_{j = 1}^{p}{s_j(i_j)} \, .$$

For notational simplicity we define by $\bu, \bv \in [d]^p$ as the index of vectors $\bX, \bY$ of $d^p$ dimensions.
Also note that $H(\bu)$ is still 2-wise independent~\citep{Patrascu11} but $S(\bu)$ is \textit{not} 4-wise independent anymore~\citep{Indyk08, Braverman10}.
That leads to the incorrect result on the variance bound on Lemma 6 of the conference version~\citep{TS} as the dependence on $p$ is missing.

Define the indicator variable $\xi_{\bu \bv} = \mathbb{I}[H(\bu) = H(\bv)]$ for any $\bu, \bv \in [d]^p$, we have

\begin{align*}
\left\langle \bCx^{(p)}, \bCy^{(p)} \right\rangle &= \sum_{\bu, \bv \in [d]^p}{\bX_{\bu} \bY_{\bv} S(\bu)S(\bv) \xi_{\bu \bv}} 
&= \left\langle \bX, \bY \right\rangle + \sum_{\bu \neq \bv \in [d]^p}{\bX_{\bu} \bY_{\bv} S(\bu)S(\bv) \xi_{\bu \bv}} \, .
\end{align*}
Recall that $S: [d]^p \rightarrow \left\{-1, 1\right\}$ is 2-wise independent, applying the independence property of this hash function, we can verify that 
$\mathbb{E}  \left[\left\langle \bCx^{(p)}, \bCy^{(p)}  \right\rangle \right] = \left\langle \bX, \bY \right\rangle = \dotxy^{p}$.

For the variance, we compute $\E{\left[\left\langle \bCx^{(p)}, \bCy^{(p)}  \right\rangle^2 \right]}$ by expanding $\left\langle \bCx^{(p)}, \bCy^{(p)} \right\rangle^2$:
\begin{align*}
&\left\langle \bCx^{(p)}, \bCy^{(p)} \right\rangle^2 = \left(  \left\langle \bX, \bY \right\rangle + \sum_{\bu \neq \bv}{\bX_{\bu} \bY_{\bv} S(\bu)S(\bv) \xi_{\bu \bv}} \right)^2 \\
&= \left\langle \bX, \bY \right\rangle^2 + 2\left\langle \bX, \bY \right\rangle \sum_{\bu \neq \bv}{\bX_{\bu} \bY_{\bv} S(\bu)S(\bv) \xi_{\bu \bv}} + \left( \sum_{\bu \neq \bv}{\bX_{\bu} \bY_{\bv} S(\bu)S(\bv) \xi_{\bu \bv}} \right)^2 \, ,
\end{align*}
where the expectation for the second term is~0.
Applying the independence between $S$ and $H$ together with Lemma~\ref{lem:AMS2}, we can bound the expectation of the last term and prove the variance claim as follows.
\begin{align*}
&\E{ \left( \sum_{\bu \neq \bv}{\bX_{\bu} \bY_{\bv} S(\bu)S(\bv) \xi_{\bu \bv}} \right)^2}  \\
&= \E{ \sum_{\substack{\bu_1 \neq \bv_1 \\ \bu_2 \neq \bv_2}} \bX_{\bu_1} \bY_{\bv_1}\bX_{\bu_2} \bY_{\bv_2} S(\bu_1)S(\bv_1)S(\bu_2)S(\bv_2) \xi_{\bu_1 \bv_1} \xi_{\bu_2 \bv_2} } \\
&= \sum_{\substack{\bu_1 \neq \bv_1 \\ \bu_2 \neq \bv_2}} \E{\bX_{\bu_1} \bY_{\bv_1}\bX_{\bu_2} \bY_{\bv_2} S(\bu_1)S(\bv_1)S(\bu_2)S(\bv_2)} \, \cdot \, \E{\xi_{\bu_1 \bv_1}\xi_{\bu_2 \bv_2}} \\
&\leq \frac{1}{D} \sum_{\substack{\bu_1 \neq \bv_1 \\ \bu_2 \neq \bv_2}} \E{\bX_{\bu_1} \bY_{\bv_1}\bX_{\bu_2} \bY_{\bv_2} S(\bu_1)S(\bv_1)S(\bu_2)S(\bv_2)}  \\
&\leq \frac{1}{D} \sum_{\substack{\bu_1 \neq \bv_1 \\ \bu_2 \neq \bv_2}} \E{|\bX_{\bu_1}| |\bY_{\bv_1}| |\bX_{\bu_2}| |\bY_{\bv_2}| S(\bu_1)S(\bv_1)S(\bu_2)S(\bv_2)} \\
&= \frac{1}{D} \, \E{\left(\sum_{\bu \neq \bv \in [d]^p} |\bX_{\bu}| \, |\bY_{\bv}| S(\bu) S(\bv)\right)^2} \leq \frac{3^p - 1}{D} \|\bX\|_2^2 \|\bY\|_2^2 = \frac{3^p - 1}{D} \|\bx\|_2^{2p}\|\by\|_2^{2p}.
\end{align*}
The last line holds using the variance bound from Lemma~\ref{lem:AMS2} over $|\bX|, |\bY|$, coordinate-wise absolute vectors of $\bX, \bY$.
Note that when $p = 1$, this general bound matches the bound of Lemma~\ref{lem:CS}.
\end{proof}

Empirically, it has been shown that normalizing a kernel may improve the performance of SVMs. 
One way to do so is to normalize the data such as $\|\bx\|_2=1$ so that the exact kernel is properly normalized, i.e. $\kappa(\bx, \bx) = \left\langle \bx, \bx \right\rangle^p = 1$. 
Theorem~\ref{thm:main} shows that Tensor Sketches can preserve the normalization of kernels given sufficiently large $D$ random features.

\section{Recent applications of Tensor Sketches}

As polynomial kernels can be used as feature representations in several computational models and applications, Tensor Sketches has emerged as a powerful tool for dimensionality reduction to scale up many computational tasks with high-dimensional datasets.
We briefly describe some recent applications that leverage Tensor Sketches as the core algorithmic component for scalability and efficiency.

\subsection*{Dimensionality Reduction.} 

Tensor Sketch has become a key tool for compressing and processing high-dimensional data, particularly in applications requiring compact representations of polynomial or multilinear feature expansions. Studies such as~\citet{Wang15, shi_multidimensional_tensor_sketch} show its use in reducing storage and computational costs in tensor-based data representations, while preserving task-relevant information. \citet{diao2018sketching} and \citet{malik2018low} explore sketching methods in low-rank approximation and matrix/tensor completion settings, where Tensor Sketch provides fast randomized projections that retain structure in the data. Additionally, works like \citet{cichocki2016tensor} and  \citet{han2020polynomial} incorporate Tensor Sketch into larger frameworks for scalable tensor factorization and deep model compression. These methods highlight Tensor Sketch’s flexibility in balancing approximation quality with substantial reductions in dimensionality and runtime.

In high-dimensional approximation tasks, Tensor Sketch offers near-linear time algorithms for polynomial kernel approximation and randomized linear algebra operations. It enables subspace embeddings and low-distortion projections with provable guarantees, as demonstrated in \citet{charikar2017hashing,ahle2020oblivious}. \citet{Avron14,song2019relative,martinsson2020randomized} further investigate sketching's role in randomized numerical linear algebra, where Tensor Sketch preserves inner products and norms in reduced dimensions. These advances underscore its power in scalable, approximate computation for large-scale scientific and machine learning applications.

\subsection*{Machine Learning and Neural Network Acceleration.}

Tensor Sketch has also seen impactful applications in accelerating machine learning and neural network pipelines, particularly in scenarios demanding high-dimensional interactions such as fine-grained classification and multimodal fusion. By approximating polynomial feature maps efficiently, Tensor Sketch enables compact bilinear pooling and interaction modeling without incurring prohibitive computational costs \citep{Fukui16, Gao16,cui2018learning}. This technique has been particularly effective in deep learning settings, where high-dimensional bilinear features are essential for tasks like few-shot learning and multimodal representation \citep{schwartz2017high, sun2020few}. Tensor Sketch is also widely used to reduce the parameter footprint in attention mechanisms and fusion layers, leading to more scalable architectures for both visual and language understanding \citep{Dai17, li2020sgm}. In addition, recent work has extended its use to kernel-based learning over structured representations, such as those arising in human action recognition and molecule generation \citep{rahmani2017learning, tripp2024tanimoto}. With further improvements in scalability and robustness \citep{fettal2023scalable, kleyko2025principled}, Tensor Sketch continues to serve as a foundational tool for compressing deep architectures, enabling neural models to scale more efficiently in resource-constrained environments.

\subsection*{Graph, Molecular and Biological Data.}

In the realm of graph and network analysis, Tensor Sketch has emerged as an effective tool for compressing high-dimensional representations of graph-structured data. By enabling efficient approximation of polynomial kernels, it supports scalable algorithms for clustering and representation learning in large networks and ordered-neighborhoods graphs~\citep{Draief18,fettal2023scalable}. These methods demonstrate the sketch’s utility in preserving structural and semantic information, making it highly suitable for tasks such as node classification, community detection, and network summarization under tight memory and computational constraints.

Tensor Sketch has proven valuable in bioinformatics and molecular machine learning, where high-dimensional and structured data representations are prevalent. Its ability to approximate polynomial kernels efficiently enables fast similarity estimation and scalable computation over molecular fingerprints and genomic data \citep{joudaki2020fast}. In molecular property prediction, \citet{tripp2024tanimoto} leverage Tensor Sketch to project molecular descriptors into compressed feature spaces that preserve interaction patterns critical for prediction tasks. These approaches illustrate how sketching techniques can maintain predictive accuracy while enabling tractable computation on large molecular databases, facilitating scalable discovery pipelines in drug design and genomics.

\section{Conclusion}\label{sec:conclusion}

In this paper, we introduce a fast and scalable sketching technique for approximating polynomial kernels, enabling efficient training of kernel-based learning algorithms.
By leveraging the connection between tensor products and the fast convolution structure of Count Sketches, our method computes random feature maps in time $\mathcal{O}(n(d + D \log D))$ for $n$ training samples in $\mathbb{R}^d$ and $D$ random features.
We also provided a theoretical analysis bounding the variance of the inner product between two Tensor Sketches, thereby establishing formal guarantees on the accuracy and reliability of the approximation.

\newpage

\acks{We thank Graham Cormode for pointing out the error in the proof of Lemma 6 in the conference version.}
\bibliography{b_Kernel,b_App,b_ConvexOpt,b_RF,b_App_RF,b_Sketch,b_Alg,b_App_TS}

\end{document}